\documentclass[a4paper]{article}
\usepackage[utf8]{inputenc}
\usepackage[T1]{fontenc}
\usepackage{graphicx}
\usepackage{amsthm,amssymb,amsfonts,amsmath}
\usepackage{xspace}
\usepackage[dvipsnames]{xcolor}
\usepackage{cleveref}
\usepackage{enumitem}

\newtheorem{theorem}{Theorem}
\newtheorem{lemma}[theorem]{Lemma}

\newtheorem{observation}[theorem]{Observation}

\newcommand{\F}{\ensuremath{\mathcal{F}}\xspace}
\newcommand{\G}{\ensuremath{\mathcal{G}}\xspace}
\newcommand{\HH}{\ensuremath{\mathcal{H}}\xspace}
\newcommand{\N}{\ensuremath{\mathbb{N}}\xspace}

\title{Avoider-Enforcer Game is NP-hard}
\author{Tillmann Miltzow\footnote{Computer Science Department Utrecht University, Netherlands. Generously supported by the Netherlands Organisation for Scientific Research (NWO) under project no. 016.Veni.192.250. {\tt t.miltzow@uu.nl}},
Milo\v{s} Stojakovi{\'c}\footnote{Department of Mathematics and Informatics, Faculty of Sciences, University of Novi Sad, Serbia. Partly supported by Ministry of Education, Science and Technological Development of the Republic of Serbia (Grant No.~451-03-68/2022-14/200125). Partly supported by Provincial Secretariat for Higher Education and Scientific Research, Province of Vojvodina (Grant No.~142-451-2686/2021). {\tt milos.stojakovic@dmi.uns.ac.rs}}}
\date{}

\begin{document}

\maketitle

\begin{abstract}
    In an Avoider-Enforcer game, we are given a hypergraph.
    Avoider and Enforcer alternate in claiming an unclaimed vertex, until all the vertices of the hypergraph are claimed. Enforcer wins if Avoider claims all vertices of an edge; Avoider wins otherwise.
    We show that it is NP-hard to decide if Avoider has a winning strategy.
\end{abstract}

\section{Introduction}

\paragraph{Motivation.}
    Positional games are a class of combinatorial games that have been extensively studied in recent literature (see books~\cite{beck2008combinatorial} and~\cite{hefetz2014positional} for an overview of the field). They include popular recreational games like Tic-Tac-Toe, Hex and Sim.
    \begin{figure}[htpb]
        \centering
        \includegraphics[scale=0.85]{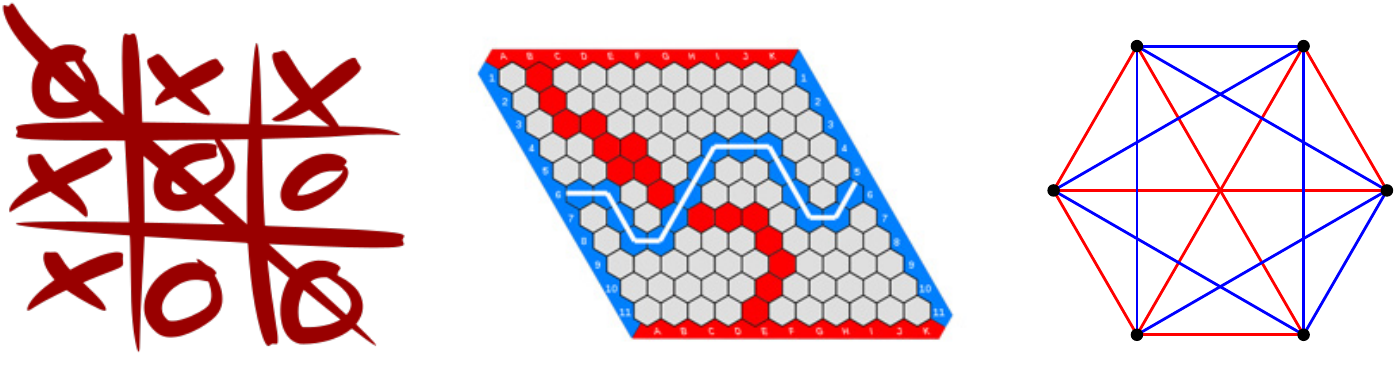}
        \caption{Recreational positional games: Tic-Tac-Toe, Hex and Sim}
    \end{figure}
    Two of the most researched types of positional games are the Maker-Breaker games and the Avoider-Enforcer games.
    While Maker-Breaker games are known to be PSPACE-complete~\cite{schaefer1978complexity, rahman20216, byskov2004maker}, even when restricted to some special games on graphs (see, e.g.,~\cite{duchene2020maker, gledel2020maker}), no classical hardness result was previously known for Avoider-Enforcer games.

    We show that it is NP-hard to determine the winner in an Avoider-Enforcer game.

\paragraph{Definitions.}
    A \textit{positional game} is a pair $(X,\F{})$, where $X$ is a finite set called a \emph{board}, and \F{} is the family of \emph{target sets}. We often refer to $(X,\F{})$ as the game hypergraph.
    The game is played by two players who alternately \textit{claim} previously unclaimed elements of $X$ until all the elements of the board are claimed.
    When it comes to the rules for determining the winner in a positional game, there are several types of games, and the two most prominent ones are Maker-Breaker games and Avoider-Enforcer games.
    In a \textit{Maker-Breaker} game, the players are called \textit{Maker} and \textit{Breaker}.
    Maker wins the game if she claims all elements of a target set from~$\mathcal{F}$, here also referred to as \textit{winning set}.
    Breaker wins otherwise, i.e.~if all elements of the board are claimed and Maker claimed no complete winning set.

    It should come as no surprise that the \textit{Avoider-Enforcer} game is played by \textit{Avoider} and \textit{Enforcer}.
    If Avoider claims all elements of a target set, which are now also called \textit{losing sets}, she loses the game, i.e., Enforcer wins. If, however, all the board elements are claimed and Avoider claimed no complete losing set, she wins the game.

\paragraph{Related Work.}
    To analyze a game, we assume that both players play optimally and then rely on the fact that one of them must have a winning strategy.
    Maker-Breaker games are first introduced by Erd\H{o}s and Selfridge~\cite{erdos1973combinatorial}, who showed that $\sum_{f\in \F }2^{-|f|}<1$  implies that Breaker has a winning strategy.
    By now Maker-Breaker games have become a well-researched class of combinatorial games,  see~\cite{hefetz2014positional} for an overview of results.
    One interesting property of Maker-Breaker games is that players do not profit from skipping moves.
    It initially may be found surprising that the analog statement does not hold for Avoider-Enforcer games --- if one or both players are offered to claim more than one board element in some moves, the outcome may change.
    That was the motivation behind the introduction of \emph{monotone Avoider-Enforcer games} in~\cite{hefetz2010avoider}, where players are allowed to claim more than one board element per move.
    To avoid confusion, the game with the original set of rules is sometimes referred to as \emph{strict Avoider-Enforcer game}.
    It turns out that both sets of rules are worth studying and the differences between them are considerable, see~\cite{hefetz2007avoider, hefetz2010avoider, bednarska2019separation, ferber2013avoider}, or see~\cite[Chapter 4]{hefetz2014positional} for an overview.

    The study of the complexity of positional games was initiated by Shaefer in~\cite{schaefer1978complexity} who shows that the problem of determining the winner of a Maker-Breaker game with winning sets of sizes up to 11 is PSPACE-complete.
    The same was recently shown to hold in~\cite{rahman20216} for Maker-Breaker games with winning sets of sizes up to six.
    A simpler proof of PSPACE-hardness of Maker-Breaker games in full generality is presented in~\cite{byskov2004maker}.

    There are numerous results on positional games played on graphs, where the board is either the edge set or the vertex set of a graph, see~\cite{hefetz2014positional} for an overview.
    As it turns out, for some of those games it is hard to determine the outcome.
    For example, the Maker-Breaker Domination Game, where the players claim the vertices of the given graph and Maker's goal is to claim a dominating set, is PSPACE-complete to resolve, even on bipartite graphs~\cite{duchene2020maker}.
    The same can be said about the Total Domination Game, a similar game where Maker wants to claim a total dominating set~\cite{gledel2020maker}.

    Let us mention here another two types of positional games, the so-called \emph{strong} games.
    In a strong Maker-Maker game the player claiming a winning set first wins, and a draw is possible. Similarly, in a strong Avoider-Avoider game a player claiming a losing set first loses the game, and again the game can end in a draw.
    The strong games are historically important, but they turn out to be notoriously hard to analyze and very few positive results about them have appeared in the literature.
    As for their complexity, it was noted in~\cite{byskov2004maker} that the PSPACE-hardness of strong Maker-Maker games can be derived in a straightforward manner from the hardness result on the Maker-Breaker games.
    Some half-played strong Avoider-Avoider games on edge sets of graphs are shown to be PSPACE-hard in~\cite{slany2002endgame}.
    \textit{Short Avoider-Enforcer} games are W[1]-complete~\cite{bonnet2017parameterized}.
    Note that the W[1]-hardness does not imply NP-hardness, as for Short Avoider-Enforcer we ask for a strategy for Avoider to survive for just $k$ steps.

     We can summarize that compared to general combinatorial games, there are very few computational complexity results in positional games.

\paragraph{Main Contribution, Discussion and Open Problems.}
    As we already mentioned, there were no previously known hardness proofs for Avoider-Enforcer games. Our main result is the following.

    \begin{theorem}
        \label{thm:strict}
        The (strict) Avoider-Enforcer game is NP-hard.
    \end{theorem}

    Note that so far we have not specified who starts the game, Avoider or Enforcer.
    It turns out that does not influence our results.

    Soon after we published the first version of this manuscript, Gledel and Oijid~\cite{gledel2022avoidance} announced their work on the same subject, improving our result by proving that strict Avoider-Enforcer games are PSPACE-complete. Furthermore, they also show how to derive the PSPACE-completeness of strong Avoider-Avoider games. Their proof is elaborate and cleverly designed to exploit a number of properties of Avoider-Enforcer games. It uses our Lemma~\ref{lem:domintaion}, and on top of that their construction for the reduction is both more complex and more technical than ours.

    Having the newly obtained results from~\cite{gledel2022avoidance} in mind, we find our main contribution to be the \emph{first} proof of hardness for Avoider-Enforcer games. On top of that, our construction and analysis are not just short but arguably also reasonably simple to follow. Finally, en route we show a general structural lemma in \Cref{sec:Monotonicity} which is of independent interest, and already applied in~\cite{gledel2022avoidance}.

    We note that all the losing sets in the proof of our \Cref{thm:strict} are of size at most six, and the same holds for the PSPACE-completeness result in~\cite{gledel2022avoidance}, so we can right away conclude that Avoider-Enforcer games are PSPACE-complete even if the losing set size is restricted to at most six.

    The only type of Avoider-Enforcer games for which there is still no known proof of hardness are the monotone Avoider-Enforcer games. This is an exciting open problem, and we conjecture that they are also PSPACE-complete in the general case.

\section{Monotonicity Lemma}
\label{sec:Monotonicity}
In this section, we prove an auxiliary result that we will later use to show the correctness of our reduction.
For stating it precisely we define a generalization of the Avoider-Enforcer game.
In this version, we are allowed to introduce cardinality constraints on subsets --
the game still has the board $X$, but the losing sets now come as pairs $(f,i_f)$, where $f\subseteq X$ and $i_f\in \N$.
Avoider loses if she has claimed at least $i_f$ vertices in the set $f$, for some set $f\in\F$.
We call this game \textit{Subset Avoider-Enforcer game}.

Note that each instance of the Subset Avoider-Enforcer game can be translated into an equivalent version of the Avoider-Enforcer game, by adding all subsets $s\subseteq f$ of size $i_f$ to the edge set.
However, this might add an exponential number of subsets.

For a Subset Avoider-Enforcer game $(X,\F)$ and a vertex $v\in X$, we denote the set $L(v):=\{(f,i_f)\in \F \mid v\in f \}$.
We say $L(a) \prec L(b)$, if for every losing tuple $(f,i_f)$ in $L(a)$ there is a losing tuple $(g,i_g)$ such that $g\subseteq f$ and $i_g\leq i_f$.

\begin{lemma}[Monotonicity Lemma]
\label{lem:domintaion}
    For a Subset Avoider-Enforcer game $(X,\F{})$, let $a,b$ be vertices such  that $L(a)\prec L(b)$.
    If a player has a winning strategy, then she also has a winning strategy in which she always prefers claiming $a$ over claiming $b$.
\end{lemma}

\begin{figure}[tbhp]
    \centering
    \includegraphics{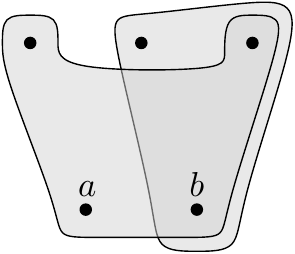}
    \caption{Intuitively, \Cref{lem:domintaion} states that if $b$ is contained in all losing sets that $a$
    is contained in, then we can assume that $a$ is preferred over $b$ by both players.}
    \label{fig:monotone}
\end{figure}

We will prove this lemma in several steps. We write $(X,\F{}) \prec (X,\G{})$
if for every losing pair $(f,i_f) \in \F{}$
there is a losing pair $(g,i_g)\in \G{} $
with $g \subseteq f$ and $i_g\leq i_f$.
Note that we use the notation $\prec$ also for hypergraphs with
different vertex sets.
In that case, we assume that there is a bijection between the vertex sets that is obvious from the context.

\begin{lemma}
    \label{lem:GraphMonotonicity}
    Let $(X,\F{}) \prec (X,\G{})$ and assume that Avoider has a winning strategy for $(X,\G{})$, then Avoider also has a winning strategy for $(X,\F{})$.
\end{lemma}
\begin{proof}
    Avoider plays in $(X,\F)$ exactly the same as she would play to win in $(X,\G)$.
    As she can avoid all sets in \G, she also avoids all sets in \F.
\end{proof}

Assume that at some point during the game play in a game $\HH = (X,\F{})$ Avoider claimed vertices $A\subseteq X$ and Enforcer claimed vertices
$E\subseteq X$, then we define the \textit{induced game} $\HH(A,E) = (Y,\G)$ as follows.
The board is $Y = X \setminus (A\cup E)$.
Furthermore, we define \G from \F as follows.
We shrink all losing sets $f$ to $g = g(f) = f\cap Y$.
We reduce $i_f$ to $i_{g}$, by reducing it by $|f\cap A|$.
In case that $i_g > |g|$, we can delete the tuple from \G.

\begin{observation}
    \label{obs:inducedGames}
    Avoider has a winning strategy in the half-played game $\HH$, after vertex sets $A$ and $E$ have been already claimed by, respectively, Avoider and Enforcer, if and only if Avoider has a winning strategy in the induced game $\HH(A,E)$.
\end{observation}

\begin{proof}[Proof of \Cref{lem:domintaion}]
    Assume that it is Avoider's turn to play.
    By \Cref{obs:inducedGames}, we can assume that this is the first move, by going to the induced game.
    Let $\HH_1  = \HH(\{a\}, \varnothing)$ and
    $\HH_2  = \HH( \{b\}, \varnothing)$.
    Then by the lemma assumption, it holds that $\HH_1 \prec \HH_2$.
    Thus by \Cref{lem:GraphMonotonicity} if Avoider has a winning strategy in $\HH_2$, then she also has a winning strategy in $\HH_1$.

    The argument for the case when it is Enforcer's turn is analog.
\end{proof}

\section{NP-hardness}
\label{sec:NP}
This section is devoted to the proof of \Cref{thm:strict}.
The idea is to reduce from 3SAT.

For that purpose, let $\varphi$ be a 3SAT instance on $n$ variables.
We construct a game $(X,\F{})$ on $4n$ vertices such that
Avoider has a winning strategy if and only if $\varphi$ is satisfiable.
We initially assume that Enforcer starts the game.
If we want Avoider to start the game, we can just add another vertex to the board that is not contained in any losing set. By \Cref{lem:domintaion}, we can assume that Avoider claims this vertex and Enforcer has the first move on the original board.

\begin{figure}[thbp]
    \centering
    \includegraphics[page = 1]{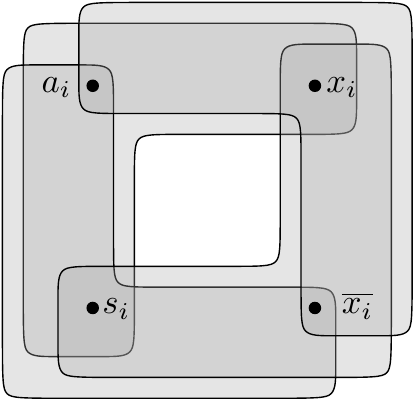}
    \caption{A group of four vertices with four triples being losing sets.}
    \label{fig:boxDefinition}
\end{figure}

\paragraph{Construction.}
In the first step of the construction, we construct $n$ groups of vertices with four vertices each,
see \Cref{fig:boxDefinition}.
We call each group a box, and we name the vertices in box $i$ by
\[a_i,s_i,x_i,\overline{x_i}.\]
This completely defines the vertex set $X$.
Note that, by a slight abuse of notation, we use the same names for vertices in $(X,\F)$ and the corresponding literals $x_i$ and $\overline{x_i}$ of $\varphi$.

We now move on to defining the losing sets~$\F{}$.
Firstly, for every box, we add all four triples of the vertices within that box to $\F{}$, see~\Cref{fig:boxDefinition}.

    Next, we construct more losing sets to encode clauses.
    We denote by $\overline{\ell}$ the negation of a literal $\ell$.
    For every clause $C=\ell_i\lor \ell_j \lor \ell_k$, where $\ell_h$ is either $x_h$ or $\overline{x_h}$, for $h=i,j,k$, we add the losing set
    $L_C=\{s_i,s_j,s_k, \overline{\ell_i}, \overline{\ell_j}, \overline{\ell_k}\}$, see \Cref{fig:example} for an illustration.

    This finishes the construction of the hypergraph $(X,\F{})$.

     \begin{figure}[bp]
        \centering
        \includegraphics[page =3]{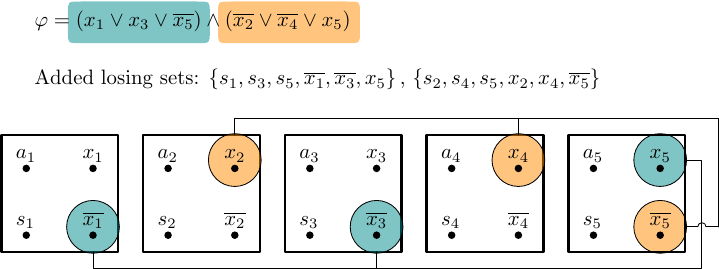}
        \caption{For each clause of $\varphi$, we construct a losing set of size six.}
        \label{fig:example}
    \end{figure}

  \bigskip

    Note that the construction can be done in linear time and space and that the collection $\F{}$ contains only losing sets with at most six vertices.
We will ultimately show that $\varphi$ is satisfiable if and only if Avoider has a winning strategy on $(X,\F{})$.
First, we need to prove two auxiliary lemmas that give some insights into the optimal strategies of Avoider and Enforcer.

\begin{lemma}
    \label{lem:box}
    Whenever Enforcer plays in one of the boxes Avoider has to play in the same box.
\end{lemma}
\begin{proof}
    Suppose for the purpose of contradiction that Enforcer plays in a box $a$ and that Avoider in the same round responds by playing in some other box. Let us further assume that this is the first time that this happens.
    Then Enforcer will keep playing in the box $a$, thus managing to claim at least three vertices of box $a$.

    As the total number of vertices of the board is even, Avoider and Enforcer will have the same numbers of claimed vertices at the end of the game.
    This implies that Avoider surely claimed at least three vertices in some other box.
    As all the triples in a box are losing sets, it follows that Avoider lost the game.
\end{proof}

\begin{lemma}
\label{lem:validStrategies}
If we restrict the game play such that within each box $i$ the vertex $a_i$ is claimed before $x_i$ and $\overline{x_i}$, and $s_i$ is claimed after $x_i$ and $\overline{x_i}$, the outcome of the game stays the same.
\end{lemma}
\begin{proof}
    Note that the triples in each box can be replaced by the cardinality constraint at most 2 in each box.
    We are now considering the corresponding Subset Avoider-Enforcer game.
    In that game we have that, for every $i$, $L(a_i) \prec L(x_i)$, $L(a_i) \prec L(\overline{x_i})$ and $L(x_i)\prec L(s_i)$, $L(\overline{x_i}) \prec L(s_i)$.
    Therefore, by \Cref{lem:domintaion}, the player who can win can do so by preferring $a_i$ over $x_i$ and $\overline{x_i}$, and by preferring $x_i$ and $\overline{x_i}$ over $s_i$. Knowing this we can restrict both players to claiming elements of the box in this order without changing the game outcome.
\end{proof}

\paragraph{Satisfiability implies Avoider wins.}
We show that the satisfiability of $\varphi$ implies a winning strategy for Avoider.
To this end, let $A$ be an assignment of all the variables that satisfy $\varphi$.

Before analyzing the game, let us restrict the game play according to the condition of \Cref{lem:validStrategies}, knowing that it does not affect the outcome of the game.
We are going to let Avoider play in the same box as Enforcer, as required by \Cref{lem:box}.
Note that under these assumptions, in each box $i$  Enforcer will claim $a_i$, then Avoider has the choice of claiming either vertex $x_i$ or vertex $\overline{x_i}$, leaving the other one to Enforcer, and finally, Avoider claims $s_i$.
Avoider will choose all $x_i$s that are true in $A$ and all $\overline{x_i}$s that are false in $A$.

We have to show that Avoider wins by following this strategy, or in other words, that she will not claim a complete losing set.
Firstly, Avoider will avoid all losing sets of size three as she will claim exactly two vertices in each box.
Furthermore, let $L_C=\{s_i,s_j,s_k, \ell_i, \ell_j, \ell_k\}$ be an arbitrary losing set encoding a clause.
Then $L_C$ corresponds to the clause
$C = \overline{\ell_i} \lor \overline{\ell_j} \lor \overline{\ell_k} $.
As $A$ is a satisfying assignment, we have at least one of the literals satisfied.
W.l.o.g.~let us assume that $\overline{\ell_i}$ is true in $A$.
Then according to the strategy Avoider has chosen vertex $\overline{\ell_i}$ and left the vertex $\ell_i$ to Enforcer,
and thus Avoider managed to avoid at least one vertex from $L_C$.
This finishes the argument for the first implication.

\paragraph{Avoider wins implies satisfiability.}
We move on to show that if Avoider has a winning strategy, then $\varphi$ is satisfiable.

We first restrict the game play according to the condition of \Cref{lem:validStrategies}. As this does not affect the outcome of the game, we know that Avoider still has a winning strategy. \Cref{lem:box} implies that, in any winning strategy, in each round, Avoider must respond in the same box in which Enforcer claims his vertex. Similarly as before we can conclude that if Avoider plays to win under the above conditions, in each box $i$ she will claim exactly one of the vertices $x_i$ and $\overline{x_i}$, and the vertex $s_i$.

For every $i$, if Avoider claimed the vertex $x_i$ we set the variable $x_i$ to true, and if Avoider claimed the vertex $\overline{x_i}$ we set the variable $x_i$ to false.
This defines $A$.

Now we have to show that $A$ is a satisfying assignment of $\varphi$.
Consider the clause $C=\ell_i\lor \ell_j \lor \ell_k$, where $\ell_h$ is either $x_h$ or $\overline{x_h}$, for $h=i,j,k$.
Then it holds that Avoider picked at least one of the vertices $\ell_i$, $\ell_j$, and $\ell_k$, as otherwise Avoider  would claim all vertices of $L_C=\{s_i,s_j,s_k, \overline{\ell_i}, \overline{\ell_j}, \overline{\ell_k}\}$ and lose the game.
Thus at least one of the literals in $C$ is true under $A$.
This shows that each clause is satisfied and thus $\varphi$ is satisfiable.

 \paragraph{Acknowledgments.}
 We would like to thank {\'E}douard Bonnet for pointers to the literature.

\bibliographystyle{plain}
\bibliography{references}

\end{document}